\documentclass[reqno]{amsart}
\usepackage{graphicx} 
\usepackage{enumerate}
\usepackage{bbm}
\usepackage{color}
\usepackage{amsaddr}
\usepackage{amsthm}
\usepackage{bm}
\usepackage{mathtools}
\usepackage{enumitem}
\usepackage{etoolbox}
\usepackage[
        doi=false,
	backend=biber,
	style=alphabetic,
	]{biblatex}
\NewBibliographyString{bywithappendix}
\DefineBibliographyStrings{english}{
  bywithappendix = {Appendix to},
}
\DeclareFieldFormat{title}{\mkbibquote{#1}}
\usepackage{hyperref}
\hypersetup{
    colorlinks=true,
    linkcolor=blue,
    citecolor=blue,
    pdfpagemode=FullScreen,
    }
\usepackage[capitalise,noabbrev,nameinlink]{cleveref}
\makeatletter
\renewcommand{\eqref}[1]{%
  \hyperref[#1]{\textup{\tagform@{\ref*{#1}}}}%
}
\makeatother
\makeatletter
\newcommand{\customlabel}[2]{%
   \protected@write \@auxout {}{\string \newlabel {#1}{{#2}{\thepage}{#2}{#1}{}} }%
   \hypertarget{#1}{}
}
\makeatother
\title[Shape of term structures]{Shape of term structures compatible with flexible choice of diffusion}

\author{Andreas Celary$^{\dag}$, Paul Kr\"uhner$^{\dag *}$}
\address{$^\dag$Institute for Statistics and Mathematics, WU-University of Economics and Business}
\email{acelary@wu.ac.at, peisenbe@wu.ac.at \normalfont{$^*$Corresponding author}}
\newtheorem{theorem}{Theorem}[section]
\newtheorem{definition}[theorem]{Definition}
\newtheorem{lemma}[theorem]{Lemma}

\newtheorem{proposition}[theorem]{Proposition}

\newtheorem{corollary}[theorem]{Corollary}
\newtheorem{remark}[theorem]{Remark}
\newtheorem{example}[theorem]{Example}

\Crefname{theorem}{Theorem}{Theorems}
\Crefname{theoremenumi}{Theorem}{Theorems}
\AtBeginEnvironment{theorem}{%
    \crefalias{enumi}{theoremenumi}%
    \setlist[enumerate,1]{
        label={\textit{\roman*)}},
        ref={\thetheorem.\roman*)}
    }%
}

\Crefname{lemma}{Lemma}{Lemmas}
\Crefname{lemmaenumi}{Lemma}{Lemmas}
\AtBeginEnvironment{lemma}{%
    \crefalias{enumi}{lemmmaenumi}%
    \setlist[enumerate,1]{
        label={\textit{\roman*)}},
        ref={\thelemma.\roman*)}
    }%
}

\Crefname{definition}{Definition}{Definitions}
\Crefname{definitionenumi}{Definition}{Definitions}
\AtBeginEnvironment{definition}{%
    \crefalias{enumi}{definitionenumi}%
    \setlist[enumerate,1]{
        label={\textit{\roman*)}},
        ref={\thedefinition.\roman*)}
    }%
}

\Crefname{remark}{Remark}{Remarks}
\Crefname{remarkenumi}{Remark}{Remarks}
\AtBeginEnvironment{remark}{%
    \crefalias{enumi}{remarkenumi}%
    \setlist[enumerate,1]{
        label={\textit{\roman*)}},
        ref={\theremark.\roman*)}
    }%
}

\newcommand{\<}{\langle}
\renewcommand{\>}{\rangle}

\newcommand{\R}{\mathbb{R}}

\DeclareMathOperator{\J}{J}   

\DeclareMathOperator{\Sg}{S}  
\DeclareMathOperator{\Tr}{Tr}  
\DeclareMathOperator{\ran}{ran}  
\DeclareMathOperator{\inn}{int}  
\DeclareMathOperator{\dom}{dom}  
\DeclareMathOperator{\rank}{rank}  
\DeclareMathOperator{\Span}{Span}  
\DeclareMathOperator{\supp}{supp}  
\DeclarePairedDelimiterX\set[1]\lbrace\rbrace{\def\given{\;\delimsize\vert\;}#1} 

\begin{document}

\begin{abstract}
   We identify all smooth manifolds of curves for Heath-Jarrow-Morton models that are consistent with any tangential diffusion coefficient. In fact, we show that these manifolds cannot be affine but must be of linear-rational type.


\end{abstract}

\maketitle

\section{Introduction}

In fixed-income markets, the basic financial instruments are securities that are based on interest rate payments. A common example among these are government bonds, which typically entail the holder to a fixed coupon, which is a regular payment, and an additional nominal at the end of the contract. The two most principle approaches to model all of the government bonds at once are the Heath-Jarrow-Morton (HJM) approach \cite{hjm} and the LIBOR market model \cite{brace_al_LIBOR}. In the HJM approach an underlying infinite-dimensional process $(f(t,T))_{0\leq t\leq T}$ is modelled, where $f(t,T)$ is called the \emph{forward rate} at time $t$ and time of maturity $T$. In this framework, the zero-coupon bond price at time $t$ with contract maturity date $T$ is defined via
\begin{equation*}
	P(t,T)=e^{ -\int _t^Tf(t,s)ds}.
\end{equation*}
 It is convenient to use (with slight abuse of notation) the parametrisation $f_t(x):=f(t,t+x)$ due to \cite{musiela} which expresses the forward rate in terms of time and time-to-maturity instead of time and maturity. One can consider the forward rate as a curve-valued object and arrive at the Heath-Jarrow-Morton-Musiela (HJMM) formulation. The original HJM-drift condition which ensures absence of arbitrage has been reformulated in \cite{musiela} into the time-to-maturity formulation.
 
 In order to make these models computable, one possibility is to set up a finite factor model $(Z_t)_{t\geq 0}$ with $f_t = g(Z_t,\cdot)$ for some deterministic function $g$. A more principled approach to this is discussed in \cite{bjoerk} where, instead of a finite-parameter set of curves, a manifold $\mathcal M$ of curves is used. \cite[p.\ 17]{bjoerk} asks three questions which are all about the relation of a forward rate curve model $(f_t)_{t\geq 0}$ and some manifold of curves $\mathcal M$, as for instance: Under which condition does a given forward rate curve model stay on a prescribed manifold of curves?  

From a more general standpoint, manifolds which remain invariant under (stochastic) dynamical systems are an important object of study as a solution space of (stochastic) partial differential equations ((S)PDEs). In mathematical finance those invariant manifolds arise naturally when considering finite-dimensional realisations of (potentially infinite-dimensional) consistent term structure models (see, e.g. \cite{bjoerk1, bjoerk2, bjoerk3, bjoerk_svensson, filip, Tappe, tappe2, tappe4}). There are several motivations for considering finite-dimensional realisations. For instance, finite dimensional models are easier simulated in a computer and if there is a finitely parametrised family of curves it is easy to statistically decide for a curve from data than in some infinite dimensional case. 

We are interested in the question of when, for a given manifold of curves $\mathcal M$, a model for the forward rate curve can be constructed. This problem was analysed in \cite{filipovic_invariant}, where equivalent conditions have been derived. In particular, any diffusion model for $f$ which is compatible with a given manifold $\mathcal M$ must have at least a diffusion coefficient which is tangential to $\mathcal M$. Additional results, as in e.g. \cite{tappe_invariant} further characterise invariant finite-dimensional manifolds for SPDEs describing the evolution of the forward curve $f$. Examples of finite-dimensional manifolds which are invariant under consistent HJMM dynamics are known. Most notably, affine term structures which are finite dimensional affine spaces have been studied in e.g. \cite{teichmann1, teich} where the authors show that the only HJMM-type models which admit a finite-dimensional realisation for essentially any initial curve $f_0$ are affine in the direction of the diffusion coefficient. In contrast to our problem specification, the authors construct invariant (finite-dimensional) manifolds given a diffusion coefficient, which turn out to be necessarily affine.


Under the assumption of a finitely parametrised set of curves $\mathcal M=\{g(z,\cdot):z\in\mathbb R^d\}$ with a smooth $g$, any $\mathbb R^d$-valued diffusion $Z$ satisfies $f_t:=g(Z_t,\cdot)\in\mathcal M$ by construction and naturally has the property that the diffusion coefficient of $f$ is tangential to $\mathcal M$. In fact, the diffusion coefficient of $f$ can be easily represented via the diffusion coefficient of $dZ_t=\beta_tdt+\sigma_tdW_t$ through an application of It\^o's formula:
\begin{equation*}
        df_t = (\partial _1g(Z_t,\cdot)\beta_t+\frac12\partial _{11}g(Z_t,\cdot)\sigma _t^2)dt + \partial_1g(Z_t,\cdot)\sigma_t dW_t. 
\end{equation*}
Given a fixed a curve fitting method for identifying the underlying forward rate curve given market data, the possible set of curves $\mathcal M=\{g(z,\cdot):z\in\mathbb R^d\}$ is typically finitely parametrised. In a second step, one might want to estimate the implicit diffusion coefficient of the parameter process from data and use the result of the estimation for a model for $Z$. 
This would require that the diffusion coefficient obtained by the method in question is indeed compatible with the given manifold $\mathcal M$. Given these considerations, our main question is this:
\begin{itemize}
    \item[(\textbf{Q})] Which smooth manifolds $\mathcal M$ of curves have consistent HJM-type models for \emph{any} possible tangential diffusion coefficient?
\end{itemize}
The precise mathematical statement is summarised in \Cref{thm:main_result} below.




The problem of calibrating financial models is well known in mathematical finance. In particular, calibration of term structure models has to be done in a careful way to preserve no-arbitrage conditions specified by the respective model. Recalibrating models frequently may lead to violations thereof. In \cite{richter_teichmann} the authors consider discrete-time recalibration schemes which generate arbitrage-free models for any time via so-called forward characteristics. Applications for term structure models of bonds and volatility curves are considered. To calibrate the term structure of interest rates, often times parametric families of curves are calibrated to properly interpolate observed prices. It is well-known (see, e.g. \cite{filipovic_consistency}) that these families have to be chosen carefully, as otherwise dynamic consistency fails to be maintained. More recent approaches have been studied, as in e.g. \cite{jarrow_wu}, where the authors attempt to employ a smoothing procedure for cubic splines to fit yield curves. To maintain dynamic consistency, the authors add higher-degree polynomials to the linear span of functions. In \cite{autoencoder}, the authors use deep-learning methods to obtain a low-dimensional representation of the term structures via autoencoders. These, however, do not necessarily stay on the invariant manifold generated by the consistent model and instead remains within a vicinity dependent on the volatility. 

In this article, we characterise all invariant finite-dimensional manifolds for the HJMM-equation which are compatible with any possible tangential diffusion coefficient. Our main result is that such a manifold must have a special rational structure. As a corollary, we find that an affine term structure where all tangential diffusion coefficients are compatible consists of a single point and implies, in fact, a deterministic model.

The article is structured as follows: \Cref{sect:prelim} provides a brief introduction to our setting in the context of invariant manifolds for interest rates. In \Cref{sect:main} we formulate our main result. The proofs for our main statements are relegated to \Cref{sect:proof} which concludes the treatment of the topic. \Cref{sect:tech_tools} contains technical tools we make use of throughout the paper. 

\section{Preliminaries}\label{sect:prelim}
\subsection{Notation}

Throughout this paper, let $(\Omega,\mathfrak A,(\mathcal F_t)_{t\geq 0},\mathbb P)$ be a complete and right-continuous filtered probability space supporting a $d$-dimensional standard Brownian motion $W$. Let $X$ and $Y$ be sets. We will denote by $\mathfrak F(X,Y)$ the space of functions from $X$ to $Y$. Furthermore, we will use the notation $\mathfrak M(X,Y)$ for the space of measurable functions from $X$ to $Y$, $\text{Lip}(X,Y)$ the space of Lipschitz-continuous functions from $X$ to $Y$, $C^k(X,Y)$ the space of $k$-times continuously (Fr\'echet) differentiable functions form $X$ to $Y$ and $L(X,Y)$ the space of linear functions from $X$ to $Y$. Let $(\mathcal{H},|\cdot|)$ be a separable Hilbert space of functions from $\mathbb R_+$ to $\mathbb R$, that is $\mathcal H\subseteq\mathfrak F(\R _+,\R)$ with
\begin{itemize}
	\item[\textbf{(H1)}] $\delta_0:\mathcal{H}\rightarrow\mathbb R, h\mapsto h(0)$ is continuous linear.\customlabel{H1}{(H1)}
	\item[\textbf{(H2)}] $\mathcal \Sg_h:\mathcal{H}\rightarrow \mathcal{H}, f\mapsto (x\mapsto f(x+h))$ defines a $c_0$-semigroup $\mathcal (\Sg_h)_{h\geq 0}$ on $\mathcal{H}$ whose generator will be denoted by $\partial_x$, cf.\ \cite[pp.\ 6--8]{EK}.\customlabel{H2}{(H2)}
\end{itemize}
We denote the scalar product on $\mathcal{H}$ by $\<h_1,h_2\>_{\mathcal H}$ for $h_1,h_2\in \mathcal{H}$. We will make use of a Hilbert space $\mathcal H$ satisfying Assumptions \ref{H1} and \ref{H2} throughout the paper without explicitly referencing them. An example of a Hilbert space with those properties is the forward curve space $\mathcal H_w$ (see \cite{filipovic_consistency}) defined as follows: 
\begin{definition}\label{def:forward_curve_space}
    Let $w:\R _+\rightarrow [1,\infty )$ be a non-decreasing $C^1$-function such that
    \begin{equation*}
        w^{-1/3}\in L^1(\R _+).
    \end{equation*}
    Define the norm
    \begin{equation*}
        \lVert h\rVert _w^2:=\lvert h(0)\rvert ^2+\int _0^{\infty}\lvert h'(x)\rvert ^2w(x)dx
    \end{equation*}
    and set
    \begin{equation*}
        \mathcal H_w:=\set{ h\in L^1_{\emph{loc}}(\R _+)\given \exists h'\in L^1_{\emph{loc}}( \R_+)\emph{ and }\lVert h\rVert _w<\infty}.
    \end{equation*}
\end{definition}
We will denote the standard unit basis vectors in $\R ^d$ as $e_i$ for $i=1,\dots,d$ and the $d$-dimensional identity matrix as $\mathbbm 1_d$. The standard scalar product on $\R^d$ will be denoted by $\langle v,w\rangle$ for $v,w\in\R^d$. We will denote elements $v$ of $\R^d$ as a vector coordinate-wise in the form $v=(v_1,\dots,v_d)$. Additionally, for $v\in\R^{d_1}$ and $w\in\R^{d_2}$ we will write the vector $u\in\R^{d_1+d_2}$ obtained by concatenating $v$ with $w$ as $u=(v,w)$. Let $X$ and $Y$ be two (topological) sets, such that $0\in Y$. Given a function $f:X\rightarrow Y$, we will denote by $\supp(f)$ the support of the function $f$, that is the smallest closed set such that $f(x)=0$ for $x\in X\setminus\supp(f)$.

\subsection{Setting}

We are interested in characterising the finite dimensional manifolds $\mathcal M\subseteq\mathcal H$ which are invariant under the HJMM-equation relative to any diffusion coefficient which is tangential to the manifold. The motivation of the problem comes from a statistical consideration. If one has decided for a set of curves $\mathcal M$ and estimates the diffusion coefficient of a process running on $\mathcal M$, then we like to be sure to obtain a model which is free of arbitrage in the sense of NAFLVR (see \cite{cuchiero_naflvr}). Indeed, it is desirable for a calibrated model to remain within the admissible space of functions no matter the values of the estimated parameters. We proceed by making mathematically clear what we mean and summarise our main finding in Theorem \ref{thm:main_result} below. We will assume that the manifold is embedded into a Hilbert space of functions which is also used to formalise the HJMM-equation.

Assume now we are given a stochastic basis $(\Omega, \mathfrak A, (\mathcal F_t)_{t\geq 0}, \mathbb P)$ and a Hilbert space $\mathcal H$ fulfilling Assumptions \ref{H1} and \ref{H2}. We begin with an observation.

\begin{remark}
	Note that the point evaluation $\delta_x:=\delta_0 \Sg_x:\mathcal{H}\rightarrow \mathbb R$ is continuous linear for any $x\geq 0.$ Let $f\in\dom(\partial_x)$. Then we have
	\begin{align*}
		\partial_xf(x) &= \left(\lim_{h\searrow0}\frac{\Sg_hf-f}{h}\right)(x) \\
                      &= \lim_{h\searrow0}\frac{f(x+h)-f(x)}{h} \\
	\end{align*} 
	where the second equality follows from the continuity of $\delta_x$. This means that $\partial_x$ is the right-derivative and its domain contains only functions that have a right derivative and where the right-derivative defines an element of $\mathcal{H}$.
\end{remark}

For the sake of completeness and to be mathematically precise, we provide a (classical) definition of manifolds we will rely on for the rest of the paper.
\begin{definition}
	Let $n\in\mathbb N$. A non-empty subset $\mathcal M\subseteq\mathcal H$ is called a $d$-dimensional manifold of class $C^k$ if for all $m\in\mathcal M$ there exists an open neighbourhood $V\subset\mathcal H$ of $m$, an open set $U\in\mathbb R^d$ and a map $\chi\in C^k(U,\mathcal H)$ such that 
	\begin{enumerate}
		\item $\chi :U\rightarrow V\cap\mathcal M$ is a homeomorphism.
		\item $D\chi (y)$ has rank $d$ for all $y\in U$.
	\end{enumerate}
	The map $\chi$ is called a \emph{local parametrisation} of $\mathcal M$. 
	We shall denote by $T_m$ the \emph{tangent space} of $\mathcal M$ at $m\in\mathcal M$, that is $T_m:=\{ D\chi (y)v: v\in\mathbb R^d\}$ for $y\in U$ so that $\chi(y) =m$. For the exact definition we refer to \cite[Section 3]{tappe_invariant_appendix}.
\end{definition}

Since we are interested in invariant manifolds for consistent interest rate models, we now formally introduce the notion of the HJMM-equation. This will serve as a basis for further considerations in the theory. Indeed, we will tacitly assume that the stochastic processes we work with in the further parts of the paper satisfy the HJMM-equation. 
\begin{definition}\label{d:HJMM-equation}
	A stochastic process $f$ with values in $\mathcal{H}$ is a solution to the HJMM-equation with diffusion operator $\Sigma : \mathcal{H} \rightarrow L(\mathbb R^d,\mathcal{H})$ if
	\begin{enumerate}
		\item $\Sigma$ is measurable and locally bounded in the sense of \cite[Definition 3.5]{tappe_existence}.
		\item $f_t = \mathcal \Sg_tf_0 + \int_0^t \Sg_{t-s}\beta(f_s) ds + \int_0^t \Sg_{t-s}\Sigma(f_s)dW_s$ for all $t\geq 0$, \label{d:HJMM-equation2},
	\end{enumerate}
	where $\beta(h):=\frac12\Sigma(h)\int_0^{\cdot}(\delta_x\Sigma(h))^* dx$ and we assume $\beta (h)\in\mathcal H$. We will say that $f$ has starting value $f_0$.
\end{definition}
\begin{remark}
	If $f$ is a solution to the HJMM-equation with diffusion operator $\Sigma$, then it is a mild solution to
	\begin{equation}\label{eq:gen_sde}
		df_t = (\partial_xf_t+\beta(f_t))dt + \Sigma(f_t)dW_t
	\end{equation}
	in the sense of \cite[Section 6.1]{zabczyk}, where $\beta$ is as in \Cref{d:HJMM-equation}.
\end{remark}
Following the arguments in \cite{zabczyk2}, the study of invariant manifolds is interesting for the sake of model calibration. Indeed, for curve fitting for interest rate models, it is desirable for the stochastic models to stay within a prescribed manifold while still staying consistent with the NAFLVR condition. This is made precise in the following definitions.
\begin{definition}\label{def:invariant_gen}
    Let $\beta:\mathcal H\rightarrow\mathcal H$ and $\Sigma:\mathcal H\rightarrow L(\R ^d, \mathcal H)$ be measurable and consider the SPDE \eqref{eq:gen_sde}
    Let $\mathcal M\subset\mathcal H$. We say $\mathcal M$ is \emph{locally invariant} under \Cref{eq:gen_sde} if for any $\tilde f_0\in\mathcal M$ there exists a continuous weak solution $f$ with lifetime $\tau$ to \eqref{eq:gen_sde} such that $f_0=\tilde f_0$ and $f_{t\wedge\tau}\in\mathcal M$ a.s.\ for all $t\geq 0$. If $\tau =\infty$, we omit the term \emph{local} in the definition above.
\end{definition}
\begin{definition}\label{def:invariant}
   Let $\mathcal M\subset\mathcal H$ and define the \emph{parameter set} $\Lambda \subset\mathfrak M(\mathcal H,\mathcal H)\times \mathfrak M(\mathcal H,L(\R^d, \mathcal H))$. We say $\mathcal M$ is \emph{fully invariant} for $\Lambda$ if $\mathcal M$ is invariant under the SPDE
    \begin{equation}\label{eq:param_sde}
        df_t = (\partial _xf_t+\beta (f_t))dt+\Sigma (f_t)dW_t
    \end{equation}
    for any pair $(\beta ,\Sigma )\in\Lambda$.
\end{definition}
%
We are mainly interested in solutions to the HJMM-equation, hence we will distinguish that special case. Furthermore, to ensure existence of solutions and sufficient regularity, we will impose some additional technical constraints on the parameter set. These are laid out in the following definition:
\begin{definition}\label{def:hjmm_invariant}
     Let $\mathcal M\subset\mathcal H$ be a submanifold and consider the HJMM-equation in \Cref{d:HJMM-equation}. Define $\mathcal B\subset C^1(\mathcal H,L(\R^d, \mathcal H))\cap\emph{Lip}(\mathcal H, L(\R^d, \mathcal H))$ to be a subset of functions $\Sigma$ with $\supp(\Sigma)\cap\mathcal M$ compact.

Define now the parameter set for the HJMM-equation as
        \begin{equation*}
            \Lambda _{\text{HJMM}} = \set[\bigg]{ (\beta, \Sigma )\given \beta(h) = \frac12\Sigma (h)\int _0^{\cdot}(\delta _x\Sigma(h))^*dx, \Sigma\in\mathcal B }.
        \end{equation*} 
    

    
    We say that $\mathcal M$ is \emph{fully invariant under the HJMM-equation} if it is fully invariant for $\Lambda _{\text{HJMM}}$.
\end{definition}
\begin{remark}\label{rmk:invariant_manifold}
	Let $\mathcal M$ be a $C^2$-manifold, $\Sigma:\mathcal{H} \rightarrow L(\mathbb R^d,\mathcal{H})$ be of class $C^1$ and $\beta:\mathcal H\rightarrow\mathcal H$ locally Lipschitz and locally bounded in the sense of \cite[Definition 2.4.2, Definition 2.4.3]{filipovic_consistency}, respectively. \cite[Theorem 3]{filipovic_invariant}, \cite{teichmann1} have shown 
    that the following statements are equivalent:
	\begin{enumerate}
		\item $\mathcal M$ is locally invariant under \Cref{eq:gen_sde} with drift $\beta$ and diffusion $\Sigma$.
		\item The following statements hold:
			\begin{enumerate}
				\item $\mathcal M\subseteq \dom(\partial_x)$. \label{rmk:invariant_manifold21}
				\item $\Sigma(m) \in L(\mathbb R^d,T_m)$ for any $m\in\mathcal M$. \label{rmk:invariant_manifold22}
				\item $\partial_xm+\beta(m)-\frac{1}{2}\sum_{j}D(\Sigma ^j)(m)\Sigma ^j(m) \in T_m$ for any $m\in\mathcal M$, where $\Sigma ^j(m):=\Sigma (m)e_j$ for $e_j\in\R^d$ the $j$-th unit basis vector of $\R^d$. \label{rmk:invariant_manifold23}
			\end{enumerate}
	\end{enumerate}
   Moreover, \cite[Theorem 2]{filipovic_invariant} yields that if $\mathcal M$ is closed in $\mathcal H$, then invariance is ensured for $\tau =\infty$.
    \cite{tappe_invariant} extend this to a setting with jump diffusions with the additional requirement of $\mathcal M$ being of class $C^3$.
\end{remark}

\begin{remark}\label{rem:factor_model}
	A natural construction of a manifold $\mathcal M$ is as follows:
	Choose $g\in C^2(\mathbb R^d,\mathcal{H})$ injective with $\rank(Dg(y))=d$ for any $y\in\mathbb R^d$ and define the $C^2$-manifold
	\begin{equation*}
		\mathcal M:=\set{g(y)\given y\in\mathbb R^d}.
	\end{equation*}
	If $f$ is a path-continuous solution to the HJMM-equation with diffusion coefficient some $\Sigma$, then it can be written as $f_t=g(Y_t)$ for some path-continuous $\mathbb R^d$-valued diffusion process $Y$. 
	For practical applications, it appears more convenient to start out with a process $Y$ with $dY_t=b_tdt+a_tdW_t$ and then define define $f_t:=g(Y_t)$. Note that in this case, one obtains from It\^o's formula 
	\begin{equation*}
		df_t = (Dg(Y_t)b_t + \frac12 D^2g(Y_t)(a_t,a_t)) dt + Dg(Y_t)a_t dW_t
	\end{equation*}
	and we see immediately that $\Sigma_t := Dg(Y_t)\in L(\mathbb R^d,T_{f_t})$ for any $t\geq 0$. If we can rewrite $\Sigma_t=\Sigma(f_t)$ for some measurable $\Sigma$, then this reads as $\Sigma(m)\in T_m$ for any $m\in\mathcal M$ (at least those $m$ which $f$ can attain). That is, the tangential condition \ref{rmk:invariant_manifold22} naturally holds from this perspective. 
	The manifold $\mathcal M:=\{g(y):y\in\mathbb R^d\}$ is fully invariant under the HJMM-equation if and only if for any possible bounded choice of diffusion coefficient $a$ with compact support for $Y$ and any starting value $Y_0$ one can find a corresponding drift coefficient $b$ such that with $dY_t=b_tdt+a_tdW_t$ one has $g(Y)$ is a solution to the HJMM-equation.
\end{remark}

\section{Main results}\label{sect:main}
In this section, we will formally state our main result, which gives a description of finite-dimensional manifolds which are fully invariant under the HJMM-equation in the sense of \Cref{def:hjmm_invariant}.
\begin{definition}\label{def:linear_rational}
	Let $\mathcal M$ be a $d$-dimensional $C^2$ manifold. We call $\mathcal M$ \emph{linear-rational} if there is an open and connected set $U\subseteq\R ^d$, and $C^2$-functions $c:\R_+\rightarrow\R$, $u:\R_+\rightarrow\R^d$ such that
	\begin{equation*}
		\mathcal M = \set[\Bigg]{  \frac{c'+\sum_{j=1}^dz_ju'_j}{1-\left(c+\sum_{j=1}^dz_ju_j\right) } \given z\in U }
	\end{equation*}
	and
	where the denominators appearing in the equation for $\mathcal{M}$ are positive, i.e. we have
	\begin{equation*}
		c(x)+\sum _{i=1}^dz_ju_j(x)<1\qquad\text{for any }x\in\R_+, z\in U.
	\end{equation*}
\end{definition}
\begin{example}
    We now give an example of a linear-rational manifold $\mathcal M$ being a $C^2$-submanifold of a Hilbert space $\mathcal H$ satisfying Assumptions \ref{H1} and \ref{H2}. Let $\alpha >0$ and choose $\lambda _1<\dots <\lambda _d<-\alpha /2$. Set $U:=(-\infty ,0)^d$ and define the functions $u_j(x):=e^{\lambda _jx}$ for $j=1,\dots,d$ and $c\equiv 0$ and let $\mathcal M$ be given as in \Cref{def:linear_rational}. For any $f\in\mathcal M$, we have for $z_1,\dots,z_d\in U$
    \begin{equation*}
        f(x)=\frac{\sum _{j=1}^dz_j\lambda _je^{\lambda _jx}}{1-\sum _{j=1}^dz_je^{\lambda _jx}}.
    \end{equation*}
    and $f\in\mathcal H_w$ for $w(x):=e^{\alpha x}$, that is, $\mathcal M$ is a $C^2$-submanifold of $\mathcal H_w$, where $\mathcal H_w$ is the forward curve space defined in \Cref{def:forward_curve_space}.
\end{example}
\begin{theorem}\label{thm:main_result}
	Let $\mathcal M$ be a connected $d$-dimensional $C^2$-manifold such that $\mathcal M$ is fully invariant under the HJMM-equation. Then $\mathcal M$ is linear-rational in the sense of \Cref{def:linear_rational} 
\end{theorem}
The proof of \Cref{thm:main_result} is provided at the end of \Cref{sect:proof}. Our main result implies that affine term structures are usually not fully invariant under the HJMM-equation. This is summarised in the next statement.
\begin{corollary}
	Let $\mathcal M\subseteq \mathcal{H}$ be a finite dimensional affine space and assume that it is fully invariant under the HJMM-equation. Then $\mathcal M$ contains a single point.
\end{corollary}
\begin{proof}
	Let $h_0, h_1\in\mathcal M$ and define the function $f:\R _+\rightarrow\mathcal H$, $t\mapsto (1-t)h_0+th_1$. Since $\mathcal M$ is an affine space, $f$ is $\mathcal M$-valued. It then follows from \Cref{lem:transport_on_manifold} that $\mathcal M$ is a singleton.
\end{proof}
\begin{remark}
    The difference between our setting and the setting in \cite{teichmann1, teich} is as follows. In \cite[Theorem 5.3]{teich} they answer the question: Given an admissible volatility structure, which HJMM-models permit a finite-dimensional realisation for essentially any initial curve $f_0$. Indeed, any invariant, finite-dimensional manifold constructed for such an initial curve $f_0$ turns out to be affine in the direction of the volatility. Our result answers a different question, namely: Which manifolds $\mathcal M$ are invariant for any tangential diffusion coefficient $\Sigma$ and any initial curve $f_0\in\mathcal M$? Here, the construction already assumes that for any $f_0\in\mathcal M$, there is a solution for the HJMM-equation with diffusion $\Sigma$ which is $\mathcal M$-valued, and our result reveals the structure of these manifolds. It turns out that these are of a linear-rational nature, cf. \Cref{thm:main_result}.
\end{remark}
The next theorem characterises the curves which appear in \Cref{thm:main_result} in the description of the manifold $\mathcal M$.
\begin{theorem}\label{thm:discount_invariant}
    Let $\mathcal H^0\subseteq \mathcal H$ be a subspace such that for any $h\in \mathcal H^0$ we have that the function $f_h(x) := h(x)\int_0^xh(u)du\in\mathcal H$ for $x\geq 0$.

	Let $\mathcal M\subseteq \mathcal H^0$ be a 
    connected $d$-dimensional $C^2$-submanifold which is closed in $\mathcal H$. Then the following statements are equivalent:
	\begin{enumerate}
		\item $\mathcal M$ is fully invariant under the HJMM-equation. 
		\item $\mathcal M$ is linear-rational in the sense of \Cref{def:linear_rational} and there is a matrix $M\in\R^{(d+1)\times (d+1)}$ such that the functions $c$ and $u$ fulfill $(c,u)(x)=(\mathbbm 1_{d+1}-e^{xM})e_1$. In particular, $c$ and $u$ are real analytic functions.
	\end{enumerate}
\end{theorem}
The proof of \Cref{thm:discount_invariant} is given at the end of \Cref{sect:proof}.

\begin{remark}
  The Hilbert space $\mathcal H_w$ given in \Cref{def:forward_curve_space} with $\mathcal H_w^0\subseteq \mathcal H_w$ given by \cite[Equation 5.12]{filipovic_consistency} satisfies that $f_h$ as given in the previous theorem is in $\mathcal H_w$ for any $h\in \mathcal H_w^0$ as shown in \cite[Lemma 5.2.1]{filipovic_consistency}.
\end{remark}

\section{Proof of the main statement}\label{sect:proof}
Throughout this section, let $\J$ denote the operator defined in \Cref{lem:j_functional}. We consider the transformation
\begin{equation}\label{eq:psi}
    \begin{aligned}
        \Psi:\ &\mathcal H \rightarrow C(\mathbb R_+,\mathbb R),\\
        &f\mapsto 1-\exp(-\J_{\cdot}f).
     \end{aligned}
\end{equation}

\begin{lemma}\label{lem:psi_properties}
    The map $\Psi$ from \Cref{eq:psi} has the following properties:
    \begin{enumerate}
        \item $\Psi f \in C^1(\mathbb R_+,\mathbb R) $ for any $f\in\mathcal H$.
        \item $(\delta_x\circ \Psi)\in C^\infty (\mathcal H,\mathbb R)$ for any $x\geq 0$.
        \item The map $\tilde \Psi:\mathcal H\times\mathbb R_+\rightarrow \mathbb R, (f,x)\mapsto (\Psi f)(x)$ satisfies $\tilde\Psi \in C^{\infty,1}(\mathcal H\times \mathbb R_+,\mathbb R)$.
	\item  $ (\Psi^{-1}h)(x) = \frac{h'(x)}{1-h(x)}$ for any $h\in\ran (\Psi )$ and any $x\geq 0$.
    \end{enumerate}
\end{lemma}
\begin{proof} Let $\Psi$ be given as in \Cref{eq:psi}.
	\begin{enumerate}
		\item Let $h\in\mathcal H$ be arbitrary. By the properties of the exponential map and the integral, we have
			\begin{equation*}
				(\Psi f)'(x)=f(x)\exp (-\J _xf),
			\end{equation*}
			which is a continuous map.
		\item The integral $\J_x$ is a linear operator on $\mathcal H$ for any $x\geq 0$, and the exponential map is smooth on $\mathcal H$, whence the statement follows.
		\item This follows directly.
		\item Let $h\in \ran(\Psi)$. Then there is $f\in \mathcal H$ with
			\begin{equation*}
				h(x) = 1-\exp(-\J_xf),\quad x\geq 0,
			\end{equation*}
			and hence
			\begin{equation*}
				\frac{h'(x)}{1-h(x)} = f(x),\quad x\geq 0
			\end{equation*}
			as required.
	\end{enumerate}
\end{proof}

If $f_t$ is the forward rate curve seen at a single point of time $t$, then the corresponding zero-coupon bond prices are given by $P(t,t+x) = \exp(-\int_0^xf_t(s)ds)$. The bond prices are also given via the discount curve $h$ by $P(t,t+x) = 1-h_t(x)$ and we see that $h_t(x) = 1-P(t,t+x) = 1-\exp(-\int_0^xf_t(s)ds)=(\Psi f_t)(x)$. Furthemore, the short rate process $r_t:=\lim _{T\rightarrow T}-\partial _T\log (P(t,T))$ fulfills $r_t=\delta _0(f_t)$. The next lemma characterises the transformation of the dynamics of $f$ into the dynamics of the transformed process $\Psi f$ in dependence of its state.
\begin{lemma}\label{lem:h_dynamics}
	Let $\Sigma: \mathcal{H}\rightarrow L(\mathcal H,\mathbb R ^d)$ be locally bounded and measurable and let $f$ be an $\mathcal{H}$-valued progressively measurable process satisfying
	\begin{equation*}
		f_t=\Sg_tf_0+\int _0^t\Sg_{t-s}\beta (f_s)ds+\int _0^t\Sg_{t-s}\Sigma (f_s)dW_s\quad\text{for any }t\geq 0,
	\end{equation*}
	where $\beta(h) :=\frac12 \Sigma (h)\J _{\cdot}\Sigma(h)^*$. Define for any $t\geq 0$ the process
	\begin{equation*}
		h_t:=\Psi f_t.
	\end{equation*}
	Then we have
	\begin{equation}\label{eq:h_dynamics1}
		h_t=S_th_0+\int _0^t\Sg_{t-s}\beta ^h(h_s)ds+\int _0^t\Sg_{t-s}\Sigma ^h(h_s)dW_s,
	\end{equation}
	where
	\begin{equation}\label{eq:h_dynamics2}
		\begin{aligned}
			\beta ^h(h)&:=(h-1)\delta _0(\Psi ^{-1}h),\\
			\Sigma ^h(h)&:=(1-h)\J_{\cdot}\Sigma (\Psi ^{-1}h)
		\end{aligned}
	\end{equation}
	Moreover, one has
	\begin{equation}\label{eq:h_dynamics3}
		\Sigma (f)(x)=\partial _x\left(\exp (\J_xf)\Sigma ^h(\Psi f)\right)
	\end{equation}
\end{lemma}
\begin{proof}
	Let $r_t=\delta _0(f_t)$ for any $t\geq 0$ and define for $0\leq t\leq T$
	\begin{equation*}
		\tilde{h}_t(T-t):=(\Psi f_t)(T-t)=1-\exp\left( -\J_{T-t}f_t\right) .
	\end{equation*}
	We have
	\begin{equation}\label{eq:h_dynamics_proof1}
		\begin{aligned}
			\J_{T-t}(f_t) &= (\J_T-\J_t)f_0 + \int_0^t (\J_{T-s}-\J_{t-s})\beta(f_s)+\int_0^t (\J_{T-s}-\J_{t-s})\Sigma(f_s)dW_s \\
                    	&= \J_Tf_0 + \int_0^t (\J_{T-s})\beta(f_s)+\int_0^t (\J_{T-s})\Sigma(f_s)dW_s\\
			&-\left(\J_tf_0+\int_0^t\J_{t-s}\beta(f_s)ds + \int_0^t\J_{t-s}\Sigma(f_s)dW_s\right).
		\end{aligned}
	\end{equation}
	Furthermore, we find that the short rate satisfies
	\begin{equation*}
		\begin{aligned}
			\int_0^t r_xdx  &=\int_0^t\left(\delta _0\Sg_xf_0 +\int_0^x\delta _0\Sg_{x-s}\beta(f_s)ds+\int_0^x\delta_0\Sg_{x-s}\Sigma(f_s)dW_s\right)dx \\
                                &= \int _0^t\delta _xf_0dx+\int _0^t\int _0^x\delta _{x-s}\beta (f_s)dsdx+\int _0^t\int _0^x\delta _{x-s}\Sigma (f_s)dW_sdx
		\end{aligned}
	\end{equation*}
	Using Fubini's theorem to exchange the order of integration we find that
	\begin{equation}\label{eq:h_dynamics_proof2}
		\begin{aligned}
                	\int _0^tr_xdx &= \int_0^tf_0(x)dx + \int_0^t\int_0^{t-s}\delta_{x}\beta(f_s)dxds + \int_0^t\int_0^{t-s}\delta_{x}\Sigma(f_s)dxdW_s \\
                	&= \J_tf_0 + \int_0^t \J_{t-s}\beta(f_s)ds+\int_0^t\J_{t-s}\Sigma(f_s)dW_s
		\end{aligned}
	\end{equation}
	Inserting identity \eqref{eq:h_dynamics_proof2} into \Cref{eq:h_dynamics_proof1} we thus have
	\begin{equation*}
		\J_{T-t}(f_t)=\J_Tf_0 + \int_0^t \J_{T-s}\beta(f_s)+\int_0^t \J_{T-s}\Sigma(f_s)dW_s-\int_0^tr_sds.
	\end{equation*}
	Consequently, we have by It\^o's formula
	\begin{equation*}
		\begin{aligned}
			\tilde{h}_t(T-t)=&\tilde{h}_0(T-t)\\
            &+\int _0^t\exp (-\J_{T-s}f_s)\left(\J _{T-s}\beta (f_s)-r_s-\frac12\J _{T-s}(\Sigma (f_s)\Sigma (f_s)^*\right)ds\\
            &+\int _0^t\exp (-\J _{T-s}f_s)\J _{T-s}\Sigma (f_s)dW_s.
		\end{aligned}
	\end{equation*}
	After substituting the expression for the HJM-drift for $\beta$, we observe after integrating that $\J _{\cdot}\beta (f)=\frac12 \J _{\cdot}\left(\Sigma (f)\Sigma(f)^*\right)$. Thus, we find
	\begin{equation*}
		\tilde{h}_t(T-t)=\tilde{h}_0(T-t)-\int _0^t\exp (-\J_{T-s}f_s)r_sds+\int _0^t\exp (\J _{T-s}f_s)\J _{T-s}\Sigma (f_s)dW_s.
	\end{equation*}
	Defining the drift- and diffusion coefficient as in \Cref{eq:h_dynamics2}, we have $\tilde{h}_t(x)=\delta_x h_t$ for any $t\geq 0$ and $x\geq 0$, where $h$ satisfies \Cref{eq:h_dynamics1} and the claim follows. Finally, \Cref{eq:h_dynamics3} follows by inverting the identity for $\Sigma ^h$ in \Cref{eq:h_dynamics2}.
\end{proof}

To proceed with our proof, we need an additional result for existence of solutions to the HJMM-equation in the special case of the regularity assumptions we have outlined in \Cref{def:hjmm_invariant}. As it turns out, this can be deduced from standard results.

\begin{lemma}\label{lem:sol_existence}
	Let $\Sigma :\mathcal H\rightarrow L(\R ^d,\mathcal H)$ be Lipschitz, such that $\emph{supp}(\Sigma)\subseteq B$ for some bounded set $B\subset\mathcal H$. Define $\beta (h):=\frac12\Sigma (h)\int _0^{\cdot}(\delta _x\Sigma (h))^*dx$. Then there is an $\mathcal H$-valued stochastic process $(f_t)_{t\geq 0}$ such that
	\begin{equation*}
		f_t=\Sg_tf_0+\int _0^t\Sg_{t-s}\beta (f_s)ds+\int _0^t\Sg_{t-s}\Sigma (f_s)dW_s,
	\end{equation*}
    for any $t\geq 0$.
\end{lemma}
\begin{proof}
First, observe that for any $h\in\mathcal H$
\begin{equation*}
    \begin{aligned}
        \beta (h)&=\frac12 \Sigma(h)\int _0^{\cdot}(\delta _x\Sigma (h))^*dx = \frac12 \Sigma(h)\int _0^{\cdot}\Sigma (h)^*\delta _x^*dx\\
        &=\frac12\Sigma (h)\Sigma (h)^*\int _0^{\cdot}\delta _x^*dx.
    \end{aligned}. 
\end{equation*}
Since $\mathcal H$ fulfills Assumption \ref{H1}, it is an RKHS and therefore there exists for all $x\in\R _+$ a function $k_x\in\mathcal H$ such that for $h\in\mathcal H$ we have $\delta _xh = \< k_x, h\> _{\mathcal H}$. We may now identify $\delta _x^*$ by using the reproducing property. We have for $h\in\mathcal H$ and $y\in\R$
\begin{equation*}
    \< h,\delta _x^*y\>_{\mathcal H}=\<\delta _xh,y\> = y\< h,k_x\>_{\mathcal H} =\<h,yk_x\>_{\mathcal H}.
\end{equation*}
Thus, $\delta _x^*$ can be identified with the function $y\mapsto yk_x$. We have $k_x=1\cdot k_x\delta _x^*(1)=\Sg^*_x\delta _0^*(1)=\Sg ^*_x(1\cdot k_0) = \Sg ^*_xk_0$. Consider for $h,h'\in\mathcal H$ the map $\phi _{h,h'}:\R_+\rightarrow\R$
\begin{equation}
    x\mapsto \< h,\Sg^*_xh'\>_{\mathcal H}=\< \Sg _xh,h'\>_{\mathcal H}. 
\end{equation}
Since $\mathcal H$ fulfills Assumption \ref{H2}, $\Sg$ is strongly continuous. This implies that the maps $\phi _{h,h'}$ are continuous for all $h,h'i\in\mathcal H$. This implies that the adjoint semigroup $(\Sg _x^*)_{x\geq 0}$ is weakly continuous. By \cite[Theorem 5.8]{engel}, $\Sg^*$ is strongly continuous (we note that more generally this holds on any reflexive Banach space) and accordingly, the map $x\mapsto k_x$ is continuous and therefore (locally) integrable. Thus, $\beta $ is well-defined.
Since $\Sigma $ is Lipschitz and with bounded support, $\beta$ is Lipschitz. \cite[Theorem 7.2]{zabczyk} yields existence of the process $(f_t)_{t\geq 0}$.
\end{proof}

Since $\mathcal M$ is fully invariant under the HJMM-equation we find that it is invariant under any $\Sigma$ with bounded support fulfilling the regularity conditions outlined in \Cref{def:hjmm_invariant}. In particular, we have $\mathcal M\subseteq \dom(\partial_x)$. This is an adaptation of the classic result
\begin{lemma}\label{l:M in dx}
	Let $\mathcal M$ be fully invariant under the HJMM-equation. Then we have
	\begin{equation*}
		\mathcal M\subseteq \dom(\partial_x).
	\end{equation*}
\end{lemma}
\begin{proof}
	Let $h\in\mathcal M$ and let $U\subseteq\mathcal H$ be an open neighbourhood of $h$. Choose $\Sigma:\mathcal{H}\rightarrow L(\mathbb R^d,\mathcal{H})$ to be of class $C^1$ with support in $U$ and $\Sigma(m)\in T_m$ for any $m\in U\cap \mathcal M$. Since $\mathcal M$ is totally invariant we find that $\mathcal M$ is invariant for the HJMM-equation with diffusion coefficient $\Sigma$. Hence, \cite[Theorem 3]{filipovic_invariant} yields that we have
\begin{equation*}
	h\in \mathcal M\cap U\subseteq \dom(\partial_x).
\end{equation*}
Since $h$ was arbitrary we have $\mathcal M\subseteq \dom(\partial_x)$.
\end{proof}
In order to prove our main result on the form of invariant manifolds, we consider here the forward rate $f$ and its transformation $h:=\Psi f$. As stated in \Cref{rem:factor_model}, it is natural to construct manifolds as a factor model by defining $f_t:=g(\cdot ,Y_t)$ with an appropriate function $g$ and $d$-dimensional diffusion Process $Y$. We will now adapt this method to the transformed process $h$ with an appropriate factor model in the form of a function $g$. The next result shows regularity results of the function $g$. In particular, we find that $g$ is sufficiently smooth for the application of the It\^o formula.
\begin{lemma}\label{l:g is C12}
	Let $U\subseteq\mathbb R^d$ be open and connected, $\chi\in C^2(U,\mathcal M)$ and assume that $\partial_x\circ\chi\in C(\inn(D),\mathcal{H})$ where $D:=\dom(\partial_x\circ\chi)$ is its domain. Define
	\begin{equation}\label{eq:g_is_c12_1}
		g(x,y) := \Psi(\chi(y))(x),\quad \text{for }x\geq 0,y\in U.
	\end{equation}
	Then $g\in C^{(1,2)}(\mathbb R_+\times U,\mathbb R)$.
\end{lemma}
\begin{proof}
	The map $T:a\mapsto 1-\exp(-a)$ is smooth. We have
	\begin{equation*}
		g(x,y) = T(\J_x(\chi(y))).
	\end{equation*}
	Consequently, $g\in C^{(1,2)}(\mathbb R_+\times U,\mathbb R)$ if $l(x,y):= \J_x(\chi(y))$ is in $C^{(1,2}(\mathbb R_+\times U,\mathbb R)$. Since convergence on $\mathcal{H}$ yields local uniform convergence due to \Cref{l:C Embedding}, we find that it is in $C^{(0,2)}(\mathbb R_+\times U,\mathbb R)$. We have $\partial_xl(x,y)=\delta _x(\chi (y))$ and hence $l\in C^{(1,2)}(\mathbb R_+\times U,\mathbb R)$.
\end{proof}
We are now prepared to use It\^o's Lemma to characterise the dynamics of the factor model $g$ in view of the dynamics of the HJMM-equation. This results in a drift condition for the invariant manifold induced by $g$ for a model fulfilling the NAFLVR condition for a model specified by a diffusion coefficient.
\begin{lemma}\label{lem:drift_cond}
	Assume the requirements of \Cref{l:g is C12} and let $g$ be the function given therein. Furthermore, let $C\subset U$ be a compact subset and $\sigma\in\mathbb R^{d\times d}$ be arbitrary. Let $\varphi\in C^2(\R^d, [0,1])$ be a bump function with $\varphi (y)=1$ for $y\in C$ and $\varphi (y)=0$ for $y\in\R^d\setminus U$ and define $\bar\sigma (y):=\varphi (y)\sigma$. There exists a measurable function $b^{\sigma}:\mathbb R^d\rightarrow \mathbb R^d, y\mapsto b^\sigma (y)$ depending on $\sigma$ such that
	\begin{equation}\label{eq:g_differential}
            \begin{aligned}
		      &\partial_xg(x,y)-\partial _xg(0,y)+\partial _xg(0,y)g(x,y)\\
              &= \langle\nabla_yg(x,y),b^\sigma (y)\rangle + \frac12 \sum_{i,j=1}^d \partial_{y_i}\partial_{y_j}g(x,y)\bar\sigma_{ij}(y)\bar\sigma _{ji}(y)
	   \end{aligned}
        \end{equation}
	for any $\sigma\in\mathbb R^{d\times d}$, $x\geq 0$ and $y\in U$.
\end{lemma}
\begin{proof}
	Assume the requirements of \Cref{l:g is C12}, which yields that $g(x,y):=\Psi (\chi (y))(x)$, $x\geq 0$, $y\in U$ is $C^{(1,2)}(\mathbb R_+\times U,\mathbb R)$. Let $\sigma\in\mathbb R^{d\times d}$ be arbitrary. Define the function $\tilde\Sigma :\mathcal H\rightarrow L(\mathbb R^d,\mathcal H)$ as
	\begin{equation*}
		\begin{aligned}
			h\mapsto\begin{cases} D\chi (\chi ^{-1}(h)) \sigma,\quad &h\in\ran (\chi),\\
				\qquad \bar\Sigma, &\text{otherwise,}
			\end{cases}
		\end{aligned}
	\end{equation*}
where $\bar\Sigma$ is an arbitrary continuous continuation of $D\chi (\chi ^{-1}(h))\sigma$. Let now $C\subset U$ be a compact set and $\varphi\in C^2(\R^d, [0,1])$ be a bump function on $C$, as specified in the statement.
Set now $\Sigma (h):=\tilde\Sigma (h)\varphi (\chi ^{-1}(h))$. It is easy to check that, by construction, $\Sigma$ satisfies the requirements of \Cref{def:hjmm_invariant}. Now, since $\mathcal M$ is fully invariant under the HJMM-equation, there is a solution $f$ to the HJMM-equation with diffusion coefficient $\Sigma$ such that $\mathbb P(f_t\in\mathcal M)=1$ for $t\geq 0$. Define now $Y_t:=\chi ^{-1}(f_t)$. By a simple stopping argument, we may restrict ourselves to the case of just a single chart $\chi$. Define $\bar\sigma (y):=\varphi (y)\sigma$. By the Inverse Function Theorem, $D\chi ^{-1}(h)=(D\chi (\chi ^{-1}(h))^{-1}$ and therefore, from It\^o's lemma we get
	\begin{equation*}
		dY_t=\beta _tdt+D\chi ^{-1}(f_t)\Sigma (f_t)dW_t=\beta_tdt+\bar\sigma (Y_t) dW_t,
	\end{equation*}
	where $\beta _t:=\text{drift}(f_t)D\chi ^{-1}(f_t)+\frac12\Tr \left(D^2\chi ^{-1}(f_t)\left( \Sigma (f_t),\Sigma (f_t)\right)\right)$. Denote the drift of the process $g(x,Y_t)$ as $\alpha _t:=\text{drift}(g(x,Y_t))$. By It\^o's lemma, $\alpha $ satisfies
	\begin{equation*}
		\alpha _t=\langle\nabla _yg(x,Y_t),\beta _t\rangle+\frac12\sum _{i,j=1}^d\partial_{y_i}\partial _{y_j}g(x,Y_t)\bar\sigma _{ij}(Y_t)\bar\sigma _{ji}(Y_t)\quad\text{for all } x,t\geq 0.
	\end{equation*}
	On the other hand, by \Cref{lem:h_dynamics}, we find that
	\begin{equation*}
		\alpha _t=\partial _xg(x,Y_t)-\partial _xg(0,Y_t)g(x,Y_t)\quad\text{for all }x,t\geq 0.
	\end{equation*}
	Define now for the given $\sigma$ the functions
	\begin{equation*}
		\begin{aligned}
			\alpha _t(y)&:=\mathbb E[\alpha _t\vert Y_t=y],\\
			b^{\sigma}_t(y)&:=\mathbb E[\beta _t\vert Y_t=y].
		\end{aligned}
	\end{equation*}
	Then, for all $x\geq 0$ and any $(y,t)\in U\times\mathbb R_+$ we have $\mathbb P^{Y_t}$-a.s.
	\begin{equation*}
	   \begin{aligned}	
            &\partial_xg(x,y)-\partial _xg(0,y)+\partial _xg(0,y)g(x,y) \\
            &= \langle\nabla_yg(x,y),b_t^\sigma (y)\rangle + \frac12 \sum_{i,j=1}^d \partial_{y_i}\partial_{y_j}g(x,y)\bar\sigma_{ij}(y)\bar\sigma _{ji}(y)
	   \end{aligned}
        \end{equation*}
	The claim now follows by the same arguments as \cite[Lemma 2.9]{shijie}.
\end{proof}
Now we can proceed by showing the affine nature of the invariant manifold for the HJMM-equation under our specifications. The next results in essence shows that locally, the image of the factor model $g$ lies in a finite-dimensional affine subspace.
\begin{lemma}\label{l:locally affine}
	Let $U\subseteq\mathbb R^d$ be open and connected and $\chi:U\rightarrow \mathcal M$, be a local map, that is\ $\chi \in C^2(U, \mathcal M)$ with $\rank(D\chi(z))=d$ for any $z\in U$. Define $G:U\rightarrow C(\mathbb R_+,\mathbb R),z\mapsto \Psi(\chi(z))$ where $\Psi$ is defined in \Cref{eq:psi}. 
Then $\ran(G)\subseteq V$ where $V$ is a $d$-dimensional affine subspace of $C(\mathbb R_+,\mathbb R)$.
\end{lemma}
The proof of the above Lemma is an adaptation of the proof in \cite[Proposition 2.14]{shijie} to our case.
\begin{proof}
	Let $C\subset U$ be a compact subset. Without loss of generality we may assume that $0\in C$. Define for $i,j=1,\dots, d$ the symmetric matrices $E_{ij}\in\R^{d\times d}$ as follows
	\begin{equation*}
		\begin{aligned}
			(E_{ij})_{kl}=\begin{cases}
				1\quad &\text{for }k=i, l=j,\\
				1\quad &\text{for }k=j, l=i,\\
				0\quad &\text{else.}
			\end{cases}
		\end{aligned}
	\end{equation*}
	For a $d\times d$-matrix $a$ and a $C^2$ bump function $\varphi$ on $C$ we define, similarly to the proof of \Cref{lem:drift_cond} the $C^1$-function
	\begin{equation*}
            \begin{aligned}
                \Sigma_a:\ &\mathcal M\rightarrow L(\mathbb R^d,\mathcal{H}), \\
                &f\mapsto \varphi (\chi ^{-1}(f))D\chi(\chi^{-1}(f))a
              \end{aligned}
	\end{equation*}
	and denote by
	\begin{equation*}
            \begin{aligned}
		  \beta_a:\ &\mathcal M\rightarrow \mathcal{H}, \\
            &f\mapsto \frac12\Sigma_a(f)\J _{\cdot}\Sigma_a(f)^*
          \end{aligned}
	\end{equation*}
	the corresponding HJM-drift coefficient. Since $\mathcal M$ is fully invariant under the HJMM-equation we have that it is invariant for solutions to the HJMM-equation with diffusion coefficient $\tilde\Sigma_a$ where $\tilde\Sigma_a$ is chosen to be any continuous continuation of $\Sigma_a$ outside $\mathcal M$.
	Define $g:\mathbb R_+\times U\rightarrow \mathbb R, g:(x,y)\mapsto G(y)(x)$. \Cref{l:g is C12} yields that $g\in C^{(1,2)}(\mathbb R_+\times U,\mathbb R)$. By \Cref{lem:drift_cond}, there is a function $b^a:\mathbb R^d\rightarrow\mathbb R ^d$, such that
	\begin{equation}\label{eq:locally_affine1}
            \begin{aligned}
		      &\partial_xg(x,y)-\partial _xg(0,y)+\partial _xg(0,y)g(x,y) \\
              & = \langle\nabla_yg(x,y),b^a (y)\rangle + \frac12 \sum_{i,j=1}^d \partial_{y_i}\partial_{y_j}g(x,y)\sigma_{ij}(y)\sigma _{ji}(y)
	   \end{aligned}
        \end{equation}
	Consider now \Cref{eq:locally_affine1} for the two cases $a =\mathbbm{1}_d$ and $a = \mathbbm{1}_d+E_{ii}$:
	\begin{equation*}
		\begin{aligned}
			\partial_x g(x,y)-\partial _xg(0,y)+\partial _xg(0,y)g(x,y) &= \langle\nabla_y g(x,y),b^{\mathbbm 1_d}(y)\rangle+\frac{1}2 \sum_{j=1}^d\partial_{y_j}^2g(x,y),\\
			\partial_x g(x,y)-\partial _xg(0,y)+\partial _xg(0,y)g(x,y) &= \langle\nabla_yg(x,y),b^{\mathbbm 1_d+E_{ii}}(y)\rangle\\
										    &+\frac32 \partial_{y_i}^2g(x,y) + \frac12 \sum_{j=1}^d\partial_{y_j}^2g(x,y),
		\end{aligned}
	\end{equation*}
	for any $x\geq 0$ and any $y\in U$. Since the left-hand-sides are identical, we then have
	\begin{equation*}
		\begin{aligned}
			&\langle\nabla_yg(x,y),b^{\mathbbm 1_d}(y)\rangle + \frac{1}2 \sum_{j=1}^d\partial_{y_j}^2g(x,y)\\
			&= \langle\nabla_yg(x,y),b^{\mathbbm 1_d+E_{ii}}(y)\rangle + \frac32 \partial_{y_i}^2g(x,y) + \frac{1}2 \sum_{j=1}^d\partial_{y_j}^2g(x,y).
		\end{aligned}
	\end{equation*}
	Hence,
	\begin{equation*}
		\partial_{y_i}^2g(x,y) = \left\langle\nabla_yg(x,y),\left(\frac23(b^{\mathbbm 1_d}-b^{\mathbbm 1_d+E_{ii}})(y)\right)\right\rangle = \langle\nabla_yg(x,y),\eta_{i,i}(y)\rangle,
	\end{equation*}
	where $\eta_{i,i}:\mathbb R^d\rightarrow \mathbb R^d$, $\eta_{i,i}:=\frac23(b^{\mathbbm 1_d}-b^{\mathbbm 1_d+E_{ii}})$. $\eta_{i,i}$ is measurable in $y$ and locally bounded by its construction as $b^{\mathbbm 1_d}$ and $b^{\mathbbm 1_d+ E_{ii}}$ are locally bounded.
	
	Repeating the same procedure for $\sigma = \mathbbm 1_d$ and $\sigma = \mathbbm 1_d + E_{ij}$ where $i,j=1,\dots,d$ and $i\neq j$, we get analogously to the previous result,
	\begin{equation*}
		\begin{aligned}
			&\langle\nabla_yg(x,y),b^{\mathbbm 1_d}(y)\rangle + \frac{1}2 \sum_{k=1}^d\partial_{y_k}^2g(x,y)\\
			&= \langle\nabla_yg(x,y),b^{\mathbbm 1_d+E_{ij}}(y)\rangle + \partial_{y_i}\partial_{y_j}g(x,y) + \frac{1}2 \sum_{k=1}^d\partial_{y_k}^2g(x,y),
		\end{aligned}
	\end{equation*}
	Hence,
	\begin{equation*}
		\partial_{y_i}\partial_{y_j}g(x,y) = \langle\nabla_yg(x,y),\left(b^{\mathbbm 1_d}-b^{\mathbbm 1_d+E_{ij}}\right) (y)\rangle = \langle\nabla_yg(x,y),\eta_{i,j}(y)\rangle,
	\end{equation*}
	where $\eta_{i,j}:= b^{\mathbbm 1_d}-b^{\mathbbm 1_d+E_{ij}}$, $\eta_{i,j}$ is measurable in $y$ and $\eta_{i,j}$ is locally bounded by its construction as $b^{\mathbbm 1_d}$ and $b^{\mathbbm 1_d+E_{ij}}$ are locally bounded.
	
	Using \Cref{eq:g_differential} with $\sigma=\mathbbm 1_d$ and multiplying it by $4$ and subtracting \Cref{eq:g_differential} with $\sigma=2\mathbbm 1_d$ yields
	\begin{equation*}
		\begin{aligned}
			&3\left(\partial_xg(x,y)-\partial_x g(0,y)+\partial _xg(0,y)g(x,y)\right)\\
            &= 4\langle\nabla_yg(x,y),b^{\mathbbm 1_d}(y)\rangle - \langle\nabla_yg(x,y),b^{2\mathbbm 1_d}(y)\rangle\\ 
												&= \langle3\nabla_yg(x,y),\gamma(y)\rangle,
		\end{aligned}
	\end{equation*}
	where $\gamma:\mathbb R^d\rightarrow \mathbb R^d, \gamma := (4b^{\mathbbm 1_d}-b^{2\mathbbm 1_d})/3$. Hence we get for any $i,j=1,\dots,d$, $x\geq 0$, $y\in U$ 
	\begin{align} 
		\partial_{y_i}\partial_{y_j}g(x,y) &= \langle\nabla_yg(x,y),\eta_{i,j}(y)\rangle, \label{eq: partial y}\\
		\partial_xg(x,y) &= \langle\nabla_yg(x,y),\gamma(y)\rangle. \label{eq:partial x}
	\end{align}
	
	Due to \Cref{eq: partial y} and \Cref{p:tcdf}, there exists a twice continuously differentiable function $A:\mathbb R^d\rightarrow \mathbb R^d$, where $A(0)=0$ such that
	\begin{equation*}
		g(x,y) = g(x,0) + \nabla_y g(x,0) A(y),\quad  x \geq 0,\ y\in U.
	\end{equation*}
	Define 
	\begin{equation*}
		V:=g(\cdot,0) + \Span\{\partial_{y_i} g(\cdot,0):i=1,\dots,d\}.
	\end{equation*}
	Then $V$ is an affine space of dimension at most $d$. Note that
	\begin{equation*}
		G(y) = g(\cdot,y) \in V,\quad y\in U.
	\end{equation*}
	Since $\rank(DG(y))=d$ we find that $I:=\{ G(y):y\in U\}$ is a $d$-dimensional $C^2$-manifold and the above shows that
	\begin{equation*}
		I\subseteq V.
	\end{equation*}
	Consequently, $\dim(V)=d$.
\end{proof}
In the next step, we utilise \Cref{l:locally affine} to show that the function $\Psi$ maps a manifold induced by a factor model into an affine subspace.
\begin{proposition}\label{p:Affine space}
	There is a $d$-dimensional affine space $V\subseteq C(\mathbb R_+,\mathbb R)$ such that
	\begin{equation*}
		\Psi(\mathcal M) \subseteq V.
	\end{equation*}
\end{proposition}
\begin{proof}
	Let $U_1,U_2$ be open and connected and $\chi_j:U_j\rightarrow \mathcal M$ be $C^2$-functions with $D\chi_j$ full rank everywhere for $j=1,2$ such that $\chi_1(U_1)\cap \chi_2(U_2)\neq \emptyset$, i.e.\ $\chi_1,\chi_2$ are local maps with a common area. Define $G_j:U\rightarrow C(\mathbb R_+,\mathbb R)$. Note that $G_1(U_1)\cap G_2(U_2)$ is non-empty and a $d$-dimensional manifold.  
	\Cref{l:locally affine} yields affine spaces $V_1,V_2\subseteq C(\mathbb R_+,\mathbb R)$ of dimension $d$ each such that $G_j(U_j)\subseteq V_j$. We find that
	\begin{equation*}
		V_1 \supseteq \left( G_1(U_1)\cap G_2(U_2)\right) \subseteq V_2.
	\end{equation*}
	Consequently, $V_1=V_2$. Since $\mathcal M$ is connected we find that $\Psi(\mathcal M)\subseteq V_1$ as claimed.
\end{proof}
We are now prepared to provide the proof for our main result.
\begin{proof}[Proof of \Cref{thm:main_result}] \label{proof:main}
	\Cref{p:Affine space} yields a $d$-dimensional affine space $V\subseteq C(\mathbb R_+,\mathbb R)$ such that
	\begin{equation*}
		\Psi(\mathcal M) \subseteq V.
	\end{equation*}
	Consequently, we have
	\begin{equation*}
		\mathcal M\subseteq \Psi^{-1}(V).
	\end{equation*}
	Let $h_0,\dots,h_d\in C(\mathbb R_+,\mathbb R)$ such that
	\begin{equation*}
		V= h_0 + \Span\{h_1,\dots,h_d\}.
	\end{equation*}
	Let $I := \Psi(\mathcal M)$. Note that $U$ is connected and open in $V$. Let $U\subseteq \mathbb R^d$ be open and connected such that the basis representation
	\begin{equation*}
		T:\mathbb R^d\rightarrow V,z\mapsto h_0 + \sum_{j=1}^dz_jh_j 
	\end{equation*}
	satisfies $T(U) = I$. We find that
	\begin{equation*}
		\mathcal M = \Psi^{-1}(I) = \Psi^{-1}(T(U)).
	\end{equation*}
	The claim follows.
\end{proof}
\begin{proof}[Proof of \Cref{thm:discount_invariant}]
    Let $\mathcal H^0$ and $\mathcal M\subseteq\mathcal H^0$ be defined as in the statement.
    \begin{itemize}
		\item[i) $\Rightarrow$ ii):] \Cref{thm:main_result} yields the asserted manifold $\mathcal M$ for some $C^2$ functions $c:\R_+\rightarrow\R$ and $u:\R _+\rightarrow\R$. Using the results of \cite[Proposition 3.2]{filipovic_discount} we obtain a matrix $M\in\R^{(d+1)\times (d+1)}$ such that $c$ and $u$ assume the asserted form.
		\item[ii) $\Rightarrow$ i):] Let $\mathcal B$ and $\Lambda _{\text{HJMM}}$ be as in \Cref{def:hjmm_invariant} and $f_0\in M$, $(\beta ,\Sigma )\in\Lambda _{\text{HJMM}}$. \Cref{lem:sol_existence} yields that there is an $\mathcal H$-valued solution $f$ to the HJMM-equation $(\beta,\Sigma)$ with starting value $f_0$. We need to show that $f$ is $\mathcal M$-valued.

        
        Consider the map $\Psi$ from \Cref{eq:psi} and define $V=\Psi (\mathcal M)$, as well as $h_t:=\Psi f_t$ for all $t\geq 0$. Define 
        \begin{equation*}
                \mathcal B^h =\set{\Sigma ^h:\mapsto (1-h)\J_{\cdot}\Sigma (\Psi ^{-1}h)\given\Sigma\in\mathcal B}
        \end{equation*}
        and consider the parameter set 
        \begin{equation*}
            \Lambda ^h:=\set{ (\beta ^h, \Sigma ^h)\given \beta ^h(h):=(h-1)\delta _0(\Psi ^{-1}h), \Sigma ^h\in\mathcal B^h}.
        \end{equation*}
        Due to \Cref{lem:h_dynamics}, if $f$ is a solution to the HJMM-equation with $(\beta ,\Sigma )\in\Lambda _{\text{HJMM}}$, then $h$ is a solution to the SPDE
        \begin{equation}\label{eq:h_dynamics_discount_invariant}
            dh_t = (\partial _xh_t+\beta^h(h_t))dt + \Sigma^h(h_t)dW_t
        \end{equation}
        with $(\beta ^h, \Sigma ^h)\in\Lambda ^h$. Indeed, since $\Psi$ is injective, in order to show that $\mathcal M$ is fully invariant under the HJMM-equation, it is sufficient to show that $V$ is fully invariant under \Cref{eq:h_dynamics_discount_invariant} for $\Lambda ^h$. We may therefore verify the conditions of \Cref{rmk:invariant_manifold} for the manifold $V$ to prove the assertion. Using the results of \Cref{p:Affine space}, we see that 
			\begin{equation*}
				V=\set[\bigg]{c+\sum _{i=1}^du_iz_i\given z\in U}.
			\end{equation*}
		Since $c$ and $u$ are real analytic functions, this easily implies $V\subseteq\text{dom}(\partial _x)$, that is \ref{rmk:invariant_manifold21}. Now, consider the pair $(\beta ^h, \Sigma ^h)\in\Lambda ^h$. Since $\J _{\cdot}$ is a linear operator, $\Sigma ^h(h)\in L(\R^d,T_h)$ for any $h\in V$, which is directly \ref{rmk:invariant_manifold22} of. To verify the third condition, we observe that $V$ is an affine manifold. Due to \cite[Theorem 4]{filipovic_invariant}, this reduces Condition \ref{rmk:invariant_manifold23} to $\partial _xh+\beta ^h(h)\in T_h$ for any $h\in T_h$. To show that this is true, define $r:=\delta _0(\Psi ^{-1}h)$. Due to \Cref{lem:psi_properties}, we find that $r=h'(0)$. Using the particular form of the functions $c$ and $u$, we therefore have $r=\langle -Me_1,(1,z)\rangle$ for $z\in U$. We may now compute
		\begin{equation*}
			\begin{aligned}
				&\partial _xh+\beta ^h(h)=\partial _xh+(h-1)r\\
							&=\partial _x\langle (\mathbbm 1_{d+1}-e^{\cdot M})e_1,(1,z)\rangle +r\left(\langle (\mathbbm 1_{d+1}-e^{\cdot M})e_1,(1,z)\rangle -1\right)\\
							&=-\langle Me^{\cdot M}e_1,(1,z)\rangle -r\langle e^{\cdot M}e_1,(1,z)\rangle\\
							&=-\langle (M-r\mathbbm 1_{d+1})e^{\cdot M}e_1,(1,z)\rangle = -\langle e^{\cdot M}e_1,(M^{\top}+\mathbbm{1}_{d+1})(1,z)\rangle\\
							&=-\langle e^{\cdot M} e_1, M^{\top}(1,z)-\langle e_1,M^{\top}(1,z)\rangle (1,z)\rangle
			\end{aligned}
		\end{equation*}
		Observe now that the vector $v:=M^{\top}(1,z)-\langle e_1, M^{\top}(1,z)\rangle(q,z)\rangle$ satisfies $v_1=M_{1,1}+\sum _{i=1}^dM_{i+1,1}z_i-(M_{1,1}-\sum _{i=1}^dM_{i+1,1}z_i)=0$. This implies $\partial _xh+\beta^h(h)\in\text{span}\{u_1,\dots,u_d\}$. Since $V$ is an affine manifold, its tangent space satisfies $T_h=V_0$, where $V_0=\text{span}\{u_1,\dots,u_d\}$, the vector space obtained by shifting the affine space to the origin. This implies that Condition \ref{rmk:invariant_manifold23} is satisfied and therefore proves the assertion.

		\end{itemize}
\end{proof}
\appendix
\section{Technical tools}\label{sect:tech_tools}
\begin{lemma}\label{l:C Embedding}
	We have $\mathcal{H}\subseteq C(\mathbb R_+,\mathbb R)$. Moreover, convergence in $\mathcal{H}$ implies local uniform convergence.

\end{lemma}
\begin{proof}
	Observe that $f(x) = \delta_0\Sg_xf$ for $x\geq 0$. Since $\Sg$ is strongly continuous we find that $x\mapsto \Sg_xf$ is continuous and, hence, $\mathcal{H}\subseteq C(\mathbb R_+,\mathbb R)$. Moreover, let $(f_n)_{n\in\mathbb N}$ be a convergent sequence in $\mathcal{H}$ with limit $f_\infty\in \mathcal{H}$. We have for $x_1\geq 0$
	\begin{equation*}
		\sup_{x\in[0,x_1]}|f_n(x)-f_\infty(x)| \leq \sup_{x\in[0,x_1]}\|\Sg_x\|\|\delta_0\||f_n-f_\infty| \rightarrow 0
	\end{equation*}
	as $n\rightarrow\infty$ because \cite[Ch.\ 1, Proposition 1.2]{EK} yields that $\sup_{x\in[0,x_1]}\|\Sg_x\| < \infty$.
\end{proof}

\begin{lemma}\label{lem:j_functional}
	Let $x\geq 0$. The linear functional $\J_x:\mathcal H\rightarrow\R$ given by
	\begin{equation*}
		\J_x(f):= \int_0^x f(s)ds
	\end{equation*}
	is defined for all $f\in \mathcal{H}$ and is continuous and linear. Moreover, the map $x\mapsto \J_x(f)$ is in $C^1(\mathbb R_+,\mathbb R)$ and
	\begin{equation*}
		\J:\mathbb R_+\times \mathcal H\rightarrow \mathbb R,\quad (x,f)\mapsto \J_x(f)
	\end{equation*}
	is continuous.
\end{lemma}
\begin{proof}
	\Cref{l:C Embedding} yields that $f\in \mathcal{H}$ is continuous. Hence, $\J_x(f)$ is well defined and $x\mapsto \J_x(f)$ is continuously differentiable.
%
	Furthermore, if $f_n\rightarrow f_\infty$ in $\mathcal{H}$ and $x_n\rightarrow x_\infty$ in $\mathbb R_+$, then $f_n\rightarrow f_\infty$ locally uniformly on $[0,\sup_{n\in\mathbb N}x_n]$ and, hence we find that
	\begin{equation*}
		\J(x_n,f_n)\rightarrow \J(x_\infty,f_\infty).
	\end{equation*}
	In particular, $\J_x$ is continuous and linear.
\end{proof}

We now provide the classical result that the generator $\partial_x$ of $\Sg$ is the derivative operator. A formal proof  is given for completeness.
\begin{lemma}\label{l:domain Generator}
	We have
	\begin{equation*}
		\dom(\partial_x) = \{f\in \mathcal{H}\cap C^1(\mathbb R_+,\mathbb R): f'\in \mathcal{H}\}
	\end{equation*}
	and $\partial_xf=f'$ for $f\in\dom(\partial_x)$, where $f'$ denotes the derivative of the function $f$.
\end{lemma}
\begin{proof}
	Let $f\in\dom(\partial_x)$. Then $\partial_xf = \lim_{h\searrow 0}\frac{\Sg_hf-f}{h}$. We find that
	\begin{equation*}
		(\partial_xf)(x) = \delta_x(\partial_xf) = \lim_{h\searrow0}\frac{f(x+h)-f(x)}{h} = f'(x+)
	\end{equation*}
	where $f'(\cdot+)$ denotes the right-derivative of $f$. Since $f'(\cdot+)\in \mathcal{H}\subseteq C(\mathbb R_+,\mathbb R)$ we find for $x>h>0$ that
	\begin{equation*}
		\frac{f(x-h)-f(x)}{-h} = \delta_{x-h}\left(\frac{\Sg_{h}f-f}{h}\right)\rightarrow \delta_x\left(\partial_xf\right) = f'(x+)
	\end{equation*}
	as $h\searrow 0$ because $\frac{\Sg_{h}f-f}{h}\rightarrow f'(\cdot+)$ locally uniformly by \Cref{l:C Embedding}. Consequently, the left-derivative equals the right-derivative and $f'(\cdot+)=f'\in \mathcal{H}\subseteq C(\mathbb R_+,\mathbb R)$. Hence,
	\begin{equation*}
		f\in \mathcal{H}\cap C^1(\mathbb R_+,\mathbb R)\text{  with  }f'\in \mathcal{H}.
	\end{equation*}
	Now, let $ f\in \mathcal{H}\cap C^1(\mathbb R_+,\mathbb R)$ with $f'\in \mathcal{H}$ and we show that $f\in\dom(\partial_x)$. We have $h\mapsto \Sg_h f'$ is continuous, hence, locally integrable. Thus, $\int_0^h \Sg_uf' du<\infty$ and we have
	\begin{equation*}
		\delta_x\left(\int_0^h \Sg_uf' du\right) = \int_0^h f'(u+x) du = f(x+h)-f(x)=\delta_x(\Sg_hf-f),
	\end{equation*}
	which shows that $\int_0^h \Sg_uf' du = \Sg_hf-f$. The fundamental theorem of calculus yields
	\begin{equation*}
		f' = \lim_{h\searrow0}\frac1h \int_0^h \Sg_uf' du  = \lim_{h\searrow0}\frac{\Sg_hf-f}{h},
	\end{equation*}
	which shows that $f\in\dom(\partial_x)$ and $\partial_xf=f'$.
\end{proof}

\begin{lemma}\label{lem:group_action}
	Let $V\subseteq\mathfrak F(\R_+,\R)$ be a finite-dimensional vector space of functions. Let $g\in V$ be such that $g^n\in V$ for any $n\in\mathbb N$. Then $g$ is constant
\end{lemma}
\begin{proof}
	Since $V$ is finite-dimensional, there is a minimal $N\in\mathbb N$, so that $g^1,\dots ,g^N$ are linearly dependent. Therefore, there exist $c_1,\dots ,c_{N-1}\in\R$, so that
	\begin{equation*}
		g^N=\sum _{i=1}^{N-1}c_ig^i.
	\end{equation*}
	Let now $\mathcal U:=\text{span}\{g^1,\dots g^{N-1}\}$. Then $\mathcal U$ is an algebra of functions. Define the linear operator $M_g:\mathcal U\rightarrow\mathcal U$, $h\mapsto gh$. Let $h$ be an eigenvector of $M_g$. We find $M_gh=\lambda h$ for some $\lambda\in\mathbb C$. This implies
	\begin{equation*}
		0 = (M_g-\lambda )h = (g-\lambda )h.
	\end{equation*}
	Thus, $\{h\neq 0\}\subseteq \{g = \lambda\}$. Since this is true for any eigenvector $h$, we have $\mathrm{ran}(g)\subseteq\text{EV}(M)$, where $\text{EV}(M)$ denotes the set of eigenvalues of the operator $M$.
\end{proof}

\begin{lemma}\label{lem:transport_on_manifold}
	Let $\mathcal M$ be linear-rational in the sense of \Cref{def:linear_rational} and $h_0, h_1\in\mathcal H$. Assume that $f: \R _+\rightarrow\mathcal H$, $t\mapsto h_0+th_1$ is $\mathcal M$-valued. Then $h_1\equiv 0$.
\end{lemma}
\begin{proof}
	Since $f$ is continuous and with values in $\mathcal M$, we find $\gamma :[0,1]\rightarrow\R ^d$ such that
	\begin{equation*}
		f(t)(x)=\frac{c'(x)+\langle u'(x),\gamma (t)\rangle}{1-\left(c(x)+\langle u(x),\gamma (t)\rangle\right)}\qquad\text{for all }t\in [0,1], x\geq 0.
	\end{equation*}
	Define now $H_0(x):=\J _xh_0$, $H_1(x):=\J _xh_1$ and we find
	\begin{equation*}
		\exp\left( -H_0(x)-tH_1(x)\right) =\exp\left( -\int _0^xf(t)(s)ds\right) = 1-\left( c(x)+\langle u(x),\gamma (t)\rangle\right)\in V,
	\end{equation*}
	where $V$ is the vector space spanned by $\{ 1-c,u_1,\dots ,u_d\}$.

	Define $g(x):=\exp\left( -H_1(x)\right)$ and the finite dimensional vector space $\mathcal U:=\{\exp\left( H_0\right) v:v\in V\}$. We observe that $g^n=\exp\left( -nH_1\right)\in\mathcal U$ by construction. \Cref{lem:group_action} yields that $g$ is constant. Consequently, $H_1$ is constant and therefore $h_1\equiv 0$ as claimed.
\end{proof}
\begin{proposition}\label{p:tcdf}
	Let $U\subseteq\mathbb R^d$ be open. Let $\eta_{i,j}:U\rightarrow \mathbb R^d$ be measurable and locally bounded with $\eta_{i,j}=\eta_{j,i}$ for any $i,j=1,\dots,d$.
	
	Then there is a twice continuously differentiable function $A:U\rightarrow \mathbb R^d$ such that for any twice continuously differentiable function $g:\mathbb R^d\rightarrow \mathbb R$ with
	\begin{equation*}
		\partial_i\partial_jg(y) = \langle\nabla g(y), \eta_{i,j}(y)\rangle,\quad i,j=1,\dots,d, y\in U, 
	\end{equation*}
	one has 
	\begin{equation*}
		g(y) = g(0) + \langle\nabla g(0), A(y)\rangle,\quad y\in\mathbb R^d. 
	\end{equation*}
\end{proposition}
\begin{proof}
	Without loss of generality we may assume that $0\in U$.

	Define 
	\begin{equation*}
		\begin{aligned}
			\mathcal S &:= \{ f\in C^2(U,\mathbb R): \partial_i\partial_j f(y) = \langle\nabla f(y), \eta_{i,j}(y)\rangle,\text{ for any } i,j=1,\dots,d, y\in U\}, \\
    			\mathcal S_0 &:= \{f\in \mathcal S: f(0)=0\}.
		\end{aligned}
	\end{equation*}
	Note that $\mathcal S$, $\mathcal S_0$ are vector spaces and for $f\in \mathcal S$ one has $h:=f-f(0)\in\mathcal S_0$ and
	\begin{equation*}
		f = f(0) + h,
	\end{equation*}
	i.e. $\mathcal S$ is the direct sum $\mathcal F_0 \oplus \mathcal S_0$ where $\mathcal F_0$ is the space of constant functions. 
	
	Let $h \in \mathcal S_0$ with $\nabla h(0) = 0$. We show that $h=0$. To this end let $R>0$. By assumption on $\eta$ there is a constant $C\geq 0$ such that $\eta_{i,j}$ is bounded on the ball with radius $R$ by $C$  for any $i,j=1,\dots,d$. We find for $x\in\mathbb R^d\setminus \{0\}$ with $|x|\leq R$ that
	\begin{equation*}
		|\nabla h(x)| = \left|\int_0^1 D(\nabla h)(tx) x dt\right| \leq Cd^2 \int_0^1 \left|\nabla h(tx)\right||x| dt. 
	\end{equation*}
	
	Thus, Gr\"onwall's lemma yields that
	\begin{equation*}
		|\nabla h(x)| \leq |\nabla h(0)| \exp( |x| Cd^2) = 0,
	\end{equation*}
	for any $x\in\mathbb R^d$ with $|x|\leq R$. Consequently, $\nabla h = 0$ which yields that $h$ is constant. Since $h(0)=0$ we find that $h=0$.
	
	Define $\Theta:\mathcal S_0\rightarrow\mathbb R^d, f\mapsto \nabla f(0)$. By the above we have that $\Theta$ is an injective linear map. Consequently, $l:= \dim(\mathcal S_0)\leq d$. Let $f_1,\dots,f_l$ be maximal linear independent in $\mathcal S_0$ such that $b^j:= \Theta(f_j)$, $j=1,\dots,l$ is an orthonormal system with respect to the standard scalar product. There is an orthogonal transformation $T:\mathbb R^d\rightarrow \mathbb R^d$ with $Tb^j = e_j$ for any $j=1,\dots, l$. We define
	\begin{equation*}
		\begin{aligned}
			f(y) :=& (f_1(y),\dots, f_l(y),0,\dots, 0), \\
			A(y) :=& T^\top (f(y)),
		\end{aligned}
	\end{equation*}
	for any $y\in\mathbb R^d$. Note that $A:\mathbb R^d\rightarrow\mathbb R^d$ is twice continuously differentiable and $A(0)=0$.
	
	Moreover, we have for $j=1,\dots,l$ and $y\in\mathbb R^d$ that
	\begin{equation*}
		f_j(0) + \langle\nabla f_j(0), A(y)\rangle = \langle b^j, T^\top (f(y))\rangle = \langle T(b^j), f(y)\rangle = \langle e_j, f(y)\rangle = f_j(y). 
	\end{equation*}
	By linearity we find that 
	\begin{equation*}
		h(0) + \langle\nabla h(0),A(y)\rangle = h(y), 
	\end{equation*}
	for any $y\in U$, $h\in\mathcal S_0$. The claim follows.
\end{proof}
\newpage
\printbibliography
\end{document}